\documentclass[a4paper]{amsart}
\usepackage{tobias-amsart}
\usepackage{verbatim}

\let\stdparagraph\paragraph
\renewcommand\paragraph{\vspace*{1em}\stdparagraph}

\renewcommand{\cite}[1]{{[\cites{#1}}}

\newcommand{\Gf}{\ensuremath{\mathfrak{G}}}

\begin{document}
\title[Regularizing Feynman path integrals using the gen. KV trace]{Regularizing Feynman path integrals using the generalized Kontsevich-Vishik trace}

\author{Tobias Hartung}
\address{Department of Mathematics, King's College London, Strand, London WC2R 2LS, United Kingdom}
\email{tobias.hartung@kcl.ac.uk}
\urladdr{www.nms.kcl.ac.uk/tobias.hartung}

\date{\today}

\begin{abstract}
  A fully regulated definition of Feynman's path integral is presented here. The proposed re-formulation of the path integral coincides with the familiar formulation whenever the path integral is well-defined. In particular, it is consistent with respect to lattice formulations and Wick rotations, i.e., it can be used in Euclidean and Minkowskian space-time. The path integral regularization is introduced through the generalized Kontsevich-Vishik trace, that is, the extension of the classical trace to Fourier Integral Operators. Physically, we are replacing the time-evolution semi-group by a holomorphic family of operator families such that the corresponding path integrals are well-defined in some half space of $\mathbb{C}$. The regularized path integral is, thus, defined through analytic continuation. This regularization can be performed by means of stationary phase approximation or computed analytically depending only on the Hamiltonian and the observable (i.e., known a priori). In either case, the computational effort to evaluate path integrals or expectations of observables reduces to the evaluation of integrals over spheres. Furthermore, computations can be performed directly in the continuum and applications (analytic computations and their implementations) to a number of models including the non-trivial cases of the massive Schwinger model and a $\phi^4$ theory.
\end{abstract}

\maketitle

\tableofcontents

\section*{Introduction}
In his original work on path integrals, Feynman~\cite{feynman} noted that recognizing known facts from different perspectives can lead to new and interesting insights. Quantum mechanics in particular has been an important example of this observation, having Schrödinger's differential equation and Heisenberg's matrix algebra. While the two theories' mathematical descriptions are seemingly distinct, Dirac's transformation theory proved their equivalence. In 1948, Feynman~\cite{feynman} added a third important mathematical formulation of quantum mechanics based on some of Dirac's observations about the role of the classical action in quantum mechanics. This third description is also known as Feynman's path integral formalism and, in combination with Feynman diagrams, proved to be fundamental for the development and study of Quantum Field Theories (QFTs).

Unfortunately, the path integral is a very elusive object. In fact, only for quantum mechanics an analytically well-defined path integral construction is known. In most other cases, the path integral can only be evaluated ``formally'', e.g., by means of a formal power series in the physical variables~\cite{johnson-freyd}. Thus, giving rise to perturbation theoretical approaches to QFT. In quantum mechanics, the path integral can be defined as a continuum limit of the discretized system~\cite{takhtajan}. Wilson~\cite{wilson} further developed this idea for QFTs since the path integral of a quantum mechanical system in discretized space-time is always well-defined. Thus, Wilson defined the path integral fully non-perturbatively on a space-time grid, going beyond perturbation theory. Using a transformation to Euclidean space-time (Wick rotation), this discretized path integral has been successfully applied to study physical systems computationally~\cite{degrand-detar,gattringer-lang,montvay-muenster} and phase space path integrals mathematically~\cite{kumano-go-I,kumano-go-II,kumano-go-III}.

Non-discretized path integrals in Euclidean space-time can be studied within the framework of classical pseudo-differential operators and their traces and determinants~\cite{paycha}. These traces and determinants are defined using $\zeta$-regularization which also gives rise to the Kontsevich-Vishik trace~\cite{kontsevich-vishik,kontsevich-vishik-geometry}. Incidentally, Hawing~\cite{hawking} had proposed studying the path integral with a curved space-time background in a $\zeta$-regularized setting long before the Kontsevich-Vishik trace was developed. In his approach, Hawking used a power series expansion of the action and regularized the quadratic term using the spectral $\zeta$-function. Furthermore, Gibbons, Hawking, and Perry~\cite{gibbons-hawking-perry} studied convergence properties of the $\zeta$-regularized one-loop approximation of the path integral.

Thus, the paper aims to shed light on the following questions.
\begin{itemize}
\item[(i)] Is it possible to $\zeta$-regularize the partition function and expectation values of observables in Minkowski space-time?
\item[(ii)] Are the regularized partition functions and expectation values of observables independent of the choices made in the construction of the $\zeta$-function?
\item[(iii)] Does the regularization contain the known special cases of well-defined path integrals (e.g., Wick rotated or space-time discretized)?
\item[(iv)] Is the construction physically ``meaningful''?
\end{itemize}

Given the recent developments on $\zeta$-functions of Fourier Integral Operators~\cite{hartung-phd,hartung-scott}, we aim to consider a non-perturbative approach to $\zeta$-regularization of path integrals. In particular, we want this new approach to contain all the special cases above, i.e., discretizations, Wick rotations, and spectral $\zeta$-functions. In order to achieve this goal, the generalized Kontsevich-Vishik trace is the ideal candidate. In fact, it can be shown that the Kontsevich-Vishik trace is the only trace on classical pseudo-differential operators (which we obtain from Wick rotating) that restricts to the canonical trace (which we obtain after discretization). Hence, we will alter Feynman's definition of the path integral to incorporate the generalized Kontsevich-Vishik trace. This ensures that the new definition of the path integral coincides with Feynman's definition whenever Feynman's path integral is well-defined.

This paper is organized as follows. Appendix~\ref{sec:KV} contains a non-technical overview of Fourier Integral Operator $\zeta$-functions and the generalized Kontsevich-Vishik trace. In section~\ref{sec:feynman}, we will use the results of appendix~\ref{sec:KV} to show that path integrals are regularizable in this sense and obtain an altered definition of the path integral, partition function, and expectation values of observables. Finally, we will consider a number of physical models in sections~\ref{sec:harmonic-oscillator}-\ref{sec:spontaneous}. First, we will give examples applying the proposed regularization to very simple models such as the harmonic oscillator, the topological oscillator, and free fermions in order to show how the regularization works in practice. In a second step, we will apply the method to non-trivial cases such as the massive Schwinger model and a $\phi^4$ theory. In particular, we will show analytic computations as well as Python implementations using symbolic arithmetic.

\subsubsection*{Acknowledgment}
The author would like to express his gratitude to Dr. Karl Jansen and Dr. Erhard Seiler for inspiring comments and conversations which helped to develop the work presented in this article.

\section{The regularized Feynman path integral}\label{sec:feynman}
Considering the Schr\"odinger equation\footnote{We use the term ``Schr\"odinger equation'' as a generic name for ``Schr\"odinger-type'' equations like the Dirac equation, i.e., we do not necessarily assume that the Hamiltonian is a Schr\"odinger operator.}
\begin{align*}
  \d_0\psi=\frac{-i}{\hbar}H\psi,
\end{align*}
we obtain
\begin{align*}
  \psi(t)=\exp\l(\frac{-i}{\hbar}\int_0^t H(s)ds\r)\psi(0).
\end{align*}
Following Feynman's approach~\cite{feynman} (cf. ``Some Remarks on Mathematical Rigor'' in chapter 4-3~\cite{feynman-hibbs-styer} and~\cite{creutz-freedman}, as well), we will change the physics slightly and introduce a (flat) time torus of length $T$, i.e.,
\begin{align*}
  \psi(T)=\exp\l(\frac{-i}{\hbar}\int_0^T H(s)ds\r)\psi(0)=\psi(0).
\end{align*}
Then, we can formally introduce the partition function
\begin{align*}
  Z_T\text{ ``$=$'' }\tr \exp\l(\frac{-i}{\hbar}\int_0^T H(s)ds\r)
\end{align*}
and the expectation of an observable $\Omega$
\begin{align*}
  \langle\Omega\rangle_T\text{ ``$=$'' }\frac{\tr\exp\l(\frac{-i}{\hbar}\int_0^T H(s)ds\r)\Omega}{Z_T}=\frac{\tr\exp\l(\frac{-i}{\hbar}\int_0^T H(s)ds\r)\Omega}{\tr\exp\l(\frac{-i}{\hbar}\int_0^T H(s)ds\r)}.
\end{align*}
The actual expectation value $\langle \Omega\rangle$ in the quantum theory can be recovered using the thermal limit
\begin{align*}
  \langle\Omega\rangle:=\lim_{T\to\infty}\langle\Omega\rangle_T.
\end{align*}
Unfortunately, the $\exp\l(\frac{-i}{\hbar}\int_0^T H(s)ds\r)\Omega$ are not of trace-class, in general. Hence, they need to be regularized.

Based on Ray and Singer's work on spectral $\zeta$-functions~\cite{ray,ray-singer}, Hawking~\cite{hawking} proposed $\zeta$-function regularization. Since most algebras of Fourier Integral Operators do not have the holomorphic functional calculus, we cannot expect to be able to define a spectral $\zeta$ function for $\exp\l(\frac{-i}{\hbar}\int_0^T H(s)ds\r)\Omega$, but we may consider (generalized) $\zeta$-functions~\cite{hartung-phd}. Thus, the regularized traces are given by the generalized Kontsevich-Vishik trace\footnote{Note that the Kontsevich-Vishik trace is the only trace on the algebra of pseudo-differential operators that coincides with the trace in $L(L_2)$ on pseudo-differential trace-class operators~\cite{maniccia-schrohe-seiler}. The generalized Kontsevich-Vishik trace is, thus, a natural choice of regularization.}~\cite{hartung-phd,hartung-scott}.

Let $H$ be a pseudo-differential operator with symbol\footnote{We call $\sigma$ the symbol of an operator $A$ if and only if $A$ is an integral operator with kernel $k$ which (locally) satisfies $k(x,y)=\int_{\rn[n]}e^{i\langle x-y,\xi\rangle_{\ell_2(n)}}\sigma(x,y,\xi)d\xi$ for some $n\in\nn$.}
\begin{align*}
  \sigma_H(t,x,r\xi):=h_2(t,x,\xi)r^2+h_1(t,x,\xi)r+h_0(t,x,r,\xi)
\end{align*}
where $\norm\xi_{\ell_2}=1$, the $h_j$ are continuous, and $h_0(t,x,r,\xi)$ has an asymptotic expansion $\sum_{j\in\nn_0}r^{-j}a_{-j}(t,x,\xi)$. Then, $\exp\l(\frac{-i}{\hbar}\int_0^T H(s)ds\r)$ has the symbol
\begin{align*}
  \sigma_{\exp\l(\frac{-i}{\hbar}\int_0^T H(s)ds\r)}=e^{iH_2(x,\xi)}e^{iH_1(x,\xi)}e^{\frac{-i}{\hbar}\int_0^T h_0(s,x,r,\xi)ds}
\end{align*}
where
\begin{align*}
  H_2(x,r\xi):=\frac{-1}{\hbar}r^2\int_0^T h_2(s,x,\xi)ds\text{ and }H_1(x,r\xi):=\frac{-1}{\hbar}r\int_0^T h_1(s,x,\xi)ds.
\end{align*}
In particular,
\begin{align*}
  e^{\frac{-i}{\hbar}\int_0^T h_0(s,x,r,\xi)ds}=&\sum_{k\in\nn_0}\frac{\l(-\frac{i}{\hbar}\r)^k}{k!}\l(\int_0^T h_0(s,x,r,\xi)ds\r)^k\\
  \sim&\sum_{k\in\nn_0}\frac{\l(-\frac{i}{\hbar}\r)^k}{k!}\l(\sum_{j\in\nn_0}r^{-j}\int_0^T a_{-j}(s,x,\xi)ds\r)^k,
\end{align*}
in combination with the power series identity (for $n\in\nn$)
\begin{align*}
  \l(\sum_{k\in\nn_0}a_kX^k\r)^n=\sum_{m\in\nn_0}c_mX^m
\end{align*}
where $c_0=a_0^n$ and $c_m=\frac{1}{ma_0}\sum_{k=1}^m(kn-m+k)a_kc_{m-k}$, shows that $e^{\frac{-i}{\hbar}\int_0^T h_0(s,x,r,\xi)ds}$ has an asymptotic expansion $b(x,\xi)\sim\sum_{j\in\nn_0}\norm\xi_{\ell_2(N)}^{-j}b_{-j}\l(x,\frac{\xi}{\norm\xi_{\ell_2(N)}}\r)$.

Regarding $\exp\l(\frac{-i}{\hbar}\int_0^T H(s)ds\r)\Omega$, we note
\begin{align*}
  \langle\Fp \Omega\phi,\Fp u\rangle=&\langle\phi,\Omega^*u\rangle\\
  =&\int_{\rn[N]}\phi(x)\l(\int_{\rn[N]}e^{i\langle x,\xi\rangle}\sigma_{\Omega^*}(x,\xi)\Fp u(\xi)d\xi\r)^*dx\\
  =&\int_{\rn[N]}\int_{\rn[N]}\phi(x)e^{-i\langle x,\xi\rangle}\sigma_{\Omega^*}(x,\xi)^*\Fp u(\xi)^*d\xi dx
\end{align*}
which implies
\begin{align*}
  \Fp\Omega\phi(\xi)=\int_{\rn[N]}e^{-i\langle x,\xi\rangle}\sigma_{\Omega^*}(x,\xi)^*\phi(x)dx
\end{align*}
and, thus,
\begin{align*}
  &\exp\l(\frac{-i}{\hbar}\int_0^T H(s)ds\r)\Omega\phi(x)\\
  =&\int_{\rn[N]}e^{iH_2(x,\xi)}e^{iH_1(x,\xi)}b(x,\xi)e^{i\langle x,\xi\rangle}\Fp\Omega\phi(\xi)d\xi\\
  =&\int_{\rn[N]}\int_{\rn[N]}e^{iH_2(x,\xi)}e^{iH_1(x,\xi)}b(x,\xi)e^{i\langle x,\xi\rangle}e^{-i\langle y,\xi\rangle}\sigma_{\Omega^*}(y,\xi)^*\phi(y)dy d\xi.
\end{align*}
Hence, (utilizing $\sigma_{A^*}(x,y,\xi)=\sigma_A(y,x,\xi)^*$ for any pseudo-differential operator $A$)
\begin{align*}
  \sigma_{\exp\l(\frac{-i}{\hbar}\int_0^T H(s)ds\r)\Omega}(x,y,\xi)=e^{iH_2(x,\xi)}e^{iH_1(x,\xi)}b(x,\xi)\sigma_{\Omega}(x,\xi).
\end{align*}
In other words, both trace integrals in $Z=\frac{\tr \exp\l(-\frac i\hbar\int H\r)\Omega}{\tr \exp\l(-\frac i\hbar\int H\r)}$ have kernels of the form
\begin{align*}
  \int_{\rn[N]}e^{i\langle x-y,\xi\rangle}e^{iH_2(x,\xi)}e^{iH_1(x,\xi)}a(x,\xi)d\xi
\end{align*}
with poly-$\log$-homogeneous $a$ provided $h_0$ and $\sigma_\Omega$ are poly-$\log$-homogeneous.

In order to $\zeta$-regularize these integrals, they need to be gauged. One of the simplest and most convenient gauges is the $\Mp$-gauge (or Mellin-gauge; cf. Definition 2.10 in~\cite{hartung-phd})
\begin{align*}
  \int_{\rn[N]}e^{i\langle x-y,\xi\rangle}e^{iH_2(x,\xi)}e^{iH_1(x,\xi)}a(x,\xi)\norm\xi_{\ell_2(N)}^z d\xi.
\end{align*}

\begin{theorem}\label{theorem-regularization}
  Let $X$ be a compact, orientable, $N$-dimensional Riemannian $C^\infty$-manifold without boundary, $\sigma_\Omega$ polyhomogeneous, and
  \begin{align*}
    Z=\int_X\int_{\rn[N]}e^{-i\sigma_H(x,\xi)}\sigma_\Omega(x,\xi)\ d\xi\ d\vol_X(x)
  \end{align*}
  with
  \begin{align*}
    \fa x\in X\ \fa r\in\rn_{\ge0}\ \fa \eta\in\d B_{\rn[N]}:\ \sigma_H(x,r\eta)=h_2(x,\eta)r^2+h_1(x,\eta)r+h_0(x,r\eta)
  \end{align*}
  where $h_2,h_1\in C(X\times\d B_{\rn[N]})$, $h_0$ polyhomogeneous, and
  \begin{enumerate}
  \item[(i)] either $h_2=0$ and $\theta(x,\xi):=h_1\l(x,\frac{\xi}{\norm\xi_{\ell_2(N)}}\r)\norm\xi_{\ell_2(N)}$ is a non-degenerate phase function
  \item[(ii)] or $\fa x\in X\ \fa\eta\in\d B_{\rn[N]}:\ \abs{h_2(x,\eta)}>0$.
  \end{enumerate}

  Then, $Z$ can be regularized using the generalized Kontsevich-Vishik trace.
\end{theorem}
\begin{proof}
  Since we can absorb $e^{ih_0}$ into the amplitude $\sigma$, we obtain without loss of generality
  \begin{align*}
    Z=\int_X\int_{\d B_{\rn[N]}}\int_{\rn_{>0}}e^{-i\l(h_2(x,\eta)r^2+h_1(x,\eta)r\r)}\sigma(x,r,\eta)\ dr\ d\vol_{\d B_{\rn[N]}}(\eta)\ d\vol_X(x).
  \end{align*}
  ``(i)'' If $h_2=0$ and $\theta(x,\xi):=h_1\l(x,\frac{\xi}{\norm\xi_{\ell_2(N)}}\r)\norm\xi_{\ell_2(N)}$ is a non-degenerate phase function, then $Z$ is a Fourier Integral Operator trace already and can, thus, be regularized using the generalized Kontsevich-Vishik trace.

  ``(ii)'' Let $\fa x\in X\ \fa\eta\in\d B_{\rn[N]}:\ \abs{h_2(x,\eta)}>0$ and
  \begin{align*}
    R:=1+\max\l\{\abs{\frac{h_1(x,\eta)}{2h_2(x,\eta)}}\in\rn;\ (x,\eta)\in X\times\d B_{\rn[N]}\r\}.
  \end{align*}
  Then, we can split $Z$ into two parts
  \begin{align*}
    Z_1:=\int_X\int_{\d B_{\rn[N]}}\int_{(0,R)}e^{-i\l(h_2(x,\eta)r^2+h_1(x,\eta)r\r)}\sigma(x,r,\eta)\ dr\ d\vol_{\d B_{\rn[N]}}(\eta)\ d\vol_X(x)
  \end{align*}
  and
  \begin{align*}
    Z_2:=\int_X\int_{\d B_{\rn[N]}}\int_{\rn_{\ge R}}e^{-i\l(h_2(x,\eta)r^2+h_1(x,\eta)r\r)}\sigma(x,r,\eta)\ dr\ d\vol_{\d B_{\rn[N]}}(\eta)\ d\vol_X(x).
  \end{align*}
  Considering $Z_1$, we observe that
  \begin{align*}
    \int_{\d B_{\rn[N]}}\int_{(0,R)}e^{-i\l(h_2(x,\eta)r^2+h_1(x,\eta)r\r)}\sigma(x,r,\eta)\ dr\ d\vol_{\d B_{\rn[N]}}(\eta)
  \end{align*}
  is the Fourier transform of a compactly supported distribution and continuous in $x$. Hence, $Z_1$ is well-defined (by Schwartz's Paley-Wiener Theorem).

  In other words, it suffices to show that we can find a Fourier Integral Operator whose trace coincides with $Z_2$ distributionally. Since $\abs{h_2}>0$, we obtain
  \begin{align*}
    h_2(x,\eta)r^2+h_1(x,\eta)r=h_2(x,\eta)\l(r+\frac{h_1(x,\eta)}{2h_2(x,\eta)}\r)^2-\frac{h_1(x,\eta)^2}{4h_2(x,\eta)}
  \end{align*}
  and, absorbing $e^{i\frac{h_1(x,\eta)^2}{4h_2(x,\eta)}}$ into $\sigma$ and setting $\sigma_1(x,s,\eta)=\sigma\l(x,s-\frac{h_1(x,\eta)}{2h_2(x,\eta)},\eta\r)$,
  \begin{align*}
    Z_2=&\int_X\int_{\d B_{\rn[N]}}\int_{\rn_{\ge R}}e^{-ih_2(x,\eta)\l(r+\frac{h_1(x,\eta)}{2h_2(x,\eta)}\r)^2}\sigma(x,r,\eta)\ dr\ d\vol_{\d B_{\rn[N]}}(\eta)\ d\vol_X(x)\\
    =&\int_X\int_{\d B_{\rn[N]}}\int_{\rn_{\ge R+\frac{h_1(x,\eta)}{2h_2(x,\eta)}}}e^{-ih_2(x,\eta)s^2}\sigma_1\l(x,s,\eta\r)\ ds\ d\vol_{\d B_{\rn[N]}}(\eta)\ d\vol_X(x)\\
    =&\int_X\int_{\d B_{\rn[N]}}\int_{\rn_{\ge\l(R+\frac{h_1(x,\eta)}{2h_2(x,\eta)}\r)^2}}e^{-ih_2(x,\eta)t}\sigma_1\l(x,\sqrt t,\eta\r)\ \frac{dt}{2\sqrt t}\ d\vol_{\d B_{\rn[N]}}(\eta)\ d\vol_X(x).
  \end{align*}
  This shows that there exists a Fourier Integral Operator with polyhomogeneous amplitude whose trace coincides with $Z_2$ (up to another Fourier transform of a compactly supported distribution).

\end{proof}

Thus, we can write
\begin{align*}
  \langle\Omega\rangle_T(z)=\frac{N(T,z)}{D(T,z)}
\end{align*}
with meromorphic functions $N(T,\cdot)$ and $D(T,\cdot)$. As we are interested in $\langle\Omega\rangle_T(0)$, there are a few cases to consider. If both $N(T,0)$ and $D(T,0)$ are regular and at most one of them vanishes, then gauge independence (cf. Lemma 2.6 in~\cite{hartung-phd}) implies that $\lim_{z\to0}\langle\Omega\rangle_T(z)$ is independent of the choice of gauge in $N$ and $D$ (though it may diverge if $D(T,0)$ vanishes). If one of the limits diverges and the other is finite, then $\lim_{z\to0}\langle\Omega\rangle_T(z)$ is either zero or divergent. Thus, the only interesting cases are if both tend to zero or diverge. In the $\frac\infty\infty$ case the result depends on the order of the pole. If the pole order of $N$ and $D$ are different, then the limits are trivial. If they are the same, then the limit is the quotient of the leading order residues, which again is gauge independent (cf. Lemma 2.5 in~\cite{hartung-phd}).

Hence, gauge dependence can only appear if $N$ or $D$ have vanishing leading Laurent coefficient or we have the $\frac00$ case. In these cases gauge dependence is, in fact, to be expected. Furthermore, the free Schwinger model (section~\ref{sec:free-schwinger}) is a $\frac00$ case, i.e., the choice of gauge is physically important. Since gauging the denominator is essentially changing physics by replacing the solution operator $\exp\l(\frac{-i}{\hbar}\int_0^t H(s)ds\r)$ with some other operator $\Gf(t,z)$, it seems sensible to apply this idea to the entire system. In other words, we are considering the family of ``evolution operators'' $\Gf(t,z)$. Then, we obtain the following new definition of our quantum theory.
\begin{definition}\label{def:reg-path-int}
  Let $H$ be the Hamiltonian, $\Omega$ an observable, and $\Gf(t,z)$ a gauged family of operators with
  \begin{align*}
    \Gf(t,0)=\exp\l(\frac{-i}{\hbar}\int_0^t H(s)ds\r).
  \end{align*}
  Then, we define the expectation value $\langle\Omega\rangle$ of $\Omega$ as
  \begin{align*}
    \langle\Omega\rangle:=\lim_{T\to\infty}\lim_{z\to0}\frac{\zeta\l(\Gf(T,\cdot)\Omega\r)(z)}{\zeta\l(\Gf(T,\cdot)\r)(z)}=\lim_{T\to\infty}\lim_{z\to0}\frac{\l.\l(s\mapsto\tr\Gf(T,s)\Omega\r)\r|_{\mathrm{mer.}}(z)}{\l.\l(s\mapsto\tr\Gf(T,s)\r)\r|_{\mathrm{mer.}}(z)}
  \end{align*}
  where $f|_{\mathrm{mer.}}$ denotes the meromorphic extension of a function $f$.
\end{definition}

\section{The harmonic oscillator}\label{sec:harmonic-oscillator}
In order to see how the proposed $\zeta$-regularization works, let us consider the traditional entry level model; the harmonic oscillator with Hamiltonian
\begin{align*}
  H=\hbar\omega\l(a^\dagger a+\frac12\r)
\end{align*}
where
\begin{align*}
  a = \sqrt{\frac{m\omega}{2\hbar}}\l(x+\frac{ip}{m\omega}\r),\ a^\dagger=\sqrt{\frac{m\omega}{2\hbar}}\l(x-\frac{ip}{m\omega}\r),\text{ and }p=-i\hbar\d.
\end{align*}
Thus,
\begin{align*}
  \sigma_a=\sqrt{\frac{m\omega}{2\hbar}}\l(x+\frac{i\hbar\xi}{m\omega}\r)\text{ and }\sigma_{a^\dagger}=\sqrt{\frac{m\omega}{2\hbar}}\l(x-\frac{i\hbar\xi}{m\omega}\r)
\end{align*}
imply
\begin{align*}
  \sigma_H=\hbar\omega\l(\sigma_{a^\dagger}\sigma_{a}+\frac12\r),\ \sigma_{\exp\l(\frac{-i}{\hbar}TH\r)} = e^{\frac{-i}{\hbar}T\sigma_H}\text{, and }\sigma_{\exp\l(\frac{-i}{\hbar}TH\r)H} = e^{\frac{-i}{\hbar}T\sigma_H}\sigma_H.
\end{align*}
Note that $\sigma_H$ is a polynomial of order $2$ in $x$. Thus, we may treat $x$ the same way we treat $\xi$ and gauge with respect to $x$, as well. Otherwise, we would have to compactify the $x$-domain and consider the limit $x$-domain$\to\rn$. Gauging in $x$ and $\xi$ yields the regularized ground state energy
\begin{align*}
  \langle H\rangle=\lim_{T\to\infty}\lim_{z_2\to0}\lim_{z_1\to0}\frac{\frac{1}{2\pi}\int_{\rn}\int_{\rn}e^{\frac{-i}{\hbar}T\sigma_H(x,\xi)}\sigma_H(x,\xi)\betr\xi^{z_1}\betr x^{z_2}d\xi dx}{\frac{1}{2\pi}\int_{\rn}\int_{\rn}e^{\frac{-i}{\hbar}T\sigma_H(x,\xi)}\betr\xi^{z_1}\betr x^{z_2}d\xi dx}.
\end{align*}
Theorem 8.7 in~\cite{hartung-phd} shows that these are trace integrals of Hilbert-Schmidt operators and regular. Although it is possible (and tedious) to compute this limit by hand, it is preferable to have a computer do the work (especially once the model is not analytically solvable anymore). Implementing this limit in Python2.7 is fairly straightforward (using the fact that integration over $\rn$ is equivalent to taking the Fourier transform and evaluating at zero).
\begin{verbatim}
import sympy as smp

z1,z2,T = smp.symbols("z1,z2,T")
x,xi = smp.symbols("x,xi",real=True)
m,hbar,omega = smp.symbols("m,hbar,omega",positive=True)


a = smp.sqrt(m*omega/(2*hbar)) * (x + smp.I*hbar*xi/(m*omega))
a_dag = smp.sqrt(m*omega/(2*hbar)) * (x - smp.I*hbar*xi/(m*omega))

h = hbar * omega * (a_dag * a + smp.sympify(1)/2)
exph = smp.exp(-smp.I * T * h / hbar)
g = smp.Abs(xi)**z1 * smp.Abs(x)**z2

num = smp.fourier_transform(h * exph * g,xi,0)
num = smp.fourier_transform(num,x,0)

den = smp.fourier_transform(exph * g,xi,0)
den = smp.fourier_transform(den,x,0)

L = smp.limit(num/den,z1,0)
L = smp.limit(L,z2,0)
print "<H> = "+str(smp.limit(L,T,smp.oo))
\end{verbatim}
This program correctly outputs the ground state energy $\langle H\rangle=\frac{\hbar\omega}{2}$.

Similarly, we can consider the $3$-dimensional harmonic oscillator whose Hamiltonian is given by
\begin{align*}
  \sigma_{H_{3D}}(x_1,x_2,x_3,\xi_1,\xi_2,\xi_3)=\sigma_{H_{1D}}(x_1,\xi_1)+\sigma_{H_{1D}}(x_2,\xi_2)+\sigma_{H_{1D}}(x_3,\xi_3).
\end{align*}
Choosing the gauge
\begin{align*}
  g(x_1,x_2,x_3,\xi_1,\xi_2,\xi_3)=\betr{\xi_1\xi_2\xi_3}^{\frac{z_1}{3}}\betr{x_1x_2x_3}^{\frac{z_2}{3}},
\end{align*}
we obtain the ground state energy $\langle H\rangle=\frac{3}{2}\hbar\omega$.

\section{The topological oscillator}\label{sec:topological-oscillator}
The topological oscillator (a.k.a. quantum rotor) models a particle of mass $M$ moving on a circle with radius $R$. Thus, choosing the angle $\phi$ as the free coordinate of the position $(x,y)=(R\cos\phi,R\sin\phi)$, we obtain the Lagrangian
\begin{align*}
  \Lp=\frac{M}{2}(\dot x^2+\dot y^2)=\frac{J}{2}\dot\phi^2
\end{align*}
with the moment of inertia $J=MR^2$. The momentum is, then, given by
\begin{align*}
  p=\d_{\dot\phi}\Lp=J\dot\phi
\end{align*}
and the Hamiltonian
\begin{align*}
  H=\dot\phi p-\Lp=\frac{p^2}{J}-\frac{p^2}{2J}=\frac{p^2}{2J}.
\end{align*}
Hence,
\begin{align*}
  \sigma_H=\frac{1}{2J}\xi^2.
\end{align*}
A characteristic value is the topological charge
\begin{align*}
  Q=&\frac{1}{2\pi}\int_0^T\dot\phi =\frac{1}{2\pi}\int_0^T\frac pJ\qquad
  \then\qquad \sigma_Q=\frac{1}{2\pi}\int_0^T\frac \xi Jdt=\frac{T\xi}{2\pi J}
\end{align*}
which counts the number of revolutions the rotor performs in the time-torus and an interesting observable is the topological susceptibility
\begin{align*}
  \chi_{\mathrm{top}}=\lim_{T\to\infty}\l\langle\frac{Q^2}{-iT}\r\rangle_T
\end{align*}
which is directly connected to the energy gap $\Delta E$ between the ground state and the first excited state
\begin{align*}
  \Delta E = 2\pi^2\chi_{\mathrm{top}}.
\end{align*}
Again, we can implement this directly in Python2.7
\begin{verbatim}
import sympy as smp

z,T = smp.symbols("z,T")
J,xi = smp.symbols("J,xi",real=True)

h = xi**2 / (2 * J)
Q = T*xi/(2 * smp.pi * J)
g = smp.Abs(xi)**z

num = smp.fourier_transform(smp.exp(-smp.I * T * h) * g * Q**2,xi,0)
den = smp.fourier_transform(smp.exp(-smp.I * T * h) * g,xi,0)

chi_top = smp.limit(smp.limit(num/(-smp.I * T * den),z,0),T,smp.oo)
energy_gap = 2 * smp.pi**2 * chi_top
print "chi_top = "+str(chi_top)
print "energy gap = "+str(energy_gap)
\end{verbatim}
and obtain the correct results $\chi_{\mathrm{top}}=\frac{1}{4\pi^2J}$ and $\Delta E=\frac{1}{2J}$.

\section{The free massive Schwinger model}\label{sec:free-schwinger}
Let us now consider the free massive Schwinger model\footnote{The massive Schwinger model can be understood as QED in two space-time dimensions.}~\cite{schwinger} whose Hamiltonian, in the zero-momentum frame using natural units $c=\hbar=1$, is given by (cf., e.g., equation~(2.2) in~\cite{jansen-cichy})
\begin{align*}
  H_m=
  \begin{pmatrix}
    m&-i\d\\
    -i\d&m
  \end{pmatrix}.
\end{align*}
Using the $\Mp$-gauge and a cut-off function $1_{B(0,X)}\le\chi\le 1_{B(0,X+1)}$ (that is, to introduce a space-torus in order to compactify the spatial domain), the ground state energy is, then, given by
\begin{align*}
  \begin{aligned}
    \langle H_m\rangle=&\lim_{X,T\to\infty}\lim_{z\to0}\frac{\frac{1}{2\pi}\int_{\rn}\int_{\rn}\chi(x)\tr\sigma_{\exp\l(-iH_mT\r)H_m}(\xi)\betr\xi^zd\xi dx}{\frac{1}{2\pi}\int_{\rn}\int_{\rn}\chi(x)\tr\sigma_{\exp\l(-iH_mT\r)}(\xi)\betr\xi^zd\xi dx}.\\
  \end{aligned}
\end{align*}
Theorem 8.7 in~\cite{hartung-phd} shows that these are trace integrals of Hilbert-Schmidt operators and regular. Implementing this limit in Python2.7 is straightforward again (the limit $X\to\infty$ can be ignored since $\langle H_m\rangle_T(z)$ is independent of $X$).
\begin{verbatim}
import sympy as smp

z = smp.symbols("z")
x,xi = smp.symbols("x,xi",real=True)
m,T = smp.symbols("m,T",positive=True)
chi = smp.Function("chi")(x)

H = smp.Matrix([[m,xi],[xi,m]])
eiTH = smp.exp(smp.I*T*H)
gauge = smp.Abs(xi)**z

num = smp.fourier_transform((H*eiTH*gauge).trace()*chi,xi,0)
den = smp.fourier_transform((eiTH*gauge).trace()*chi,xi,0)
num = smp.integrate(num,x)
den = smp.integrate(den,x)

print "<H_m> = "+str(smp.limit(smp.limit(num/den,z,0),T,smp.oo))
\end{verbatim}
This program outputs \verb|<H_m> = m|. In other words, we have just correctly computed
\begin{align*}
  E=mc^2
\end{align*}
for the free massive Schwinger model.

\section{Free relativistic Fermions}\label{sec:free-relativistic-Fermions}
Let us now step up to $4$ space-time dimensions and consider a free relativistic fermion of mass $m$. Then, using the Pauli matrices $\sigma_k$, we obtain the Hamiltonian (Einstein summation over spatial indices)
\begin{align*}
  H_m=
  \begin{pmatrix}
    m&-i\sigma_k\d_k\\
    -i\sigma_k\d_k&m
  \end{pmatrix}
\end{align*}
which yields
\begin{align*}
  \sigma_{\exp(-iH_mT)}=&e^{-imT}\l(\cos\l(T\norm\xi_{\ell_2(3)}\r)-\frac{i\sin\l(T\norm\xi_{\ell_2(3)}\r)}{\norm\xi_{\ell_2(3)}}
  \begin{pmatrix}
    0&\sigma_k\xi_k\\
    \sigma_k\xi_k&0
  \end{pmatrix}
  \r)
\end{align*}
and
\begin{align*}
  \sigma_{\exp(-iH_mT)H_m}=&e^{-imT}\l(\cos\l(T\norm\xi_{\ell_2(3)}\r)
  \begin{pmatrix}
    m&\sigma_k\xi_k\\
    \sigma_k\xi_k&m
  \end{pmatrix}
  -i\norm\xi_{\ell_2(3)}\sin\l(T\norm\xi_{\ell_2(3)}\r)\r).
\end{align*}
Thus,
\begin{align*}
  \langle H_m\rangle=&\lim_{T\to\infty}\lim_{z\to0}\frac{\int_{\rn[3]}\l(4m\cos\l(T\norm\xi_{\ell_2(3)}\r)-4i\norm\xi_{\ell_2(3)}\sin\l(T\norm\xi_{\ell_2(3)}\r)\r)\norm\xi_{\ell_2(3)}^zd\xi}{\int_{\rn[3]}4\cos\l(T\norm\xi_{\ell_2(3)}\r)\norm\xi_{\ell_2(3)}^zd\xi}\\
  =&m+\lim_{T\to\infty}\lim_{z\to0}\frac{-i\int_{\rn[3]}\norm\xi_{\ell_2(3)}^{z+1}\l(e^{iT\norm\xi_{\ell_2(3)}}-e^{-iT\norm\xi_{\ell_2(3)}}\r)d\xi}{\int_{\rn[3]}\norm\xi_{\ell_2(3)}^{z}\l(e^{iT\norm\xi_{\ell_2(3)}}+e^{-iT\norm\xi_{\ell_2(3)}}\r)d\xi}\\
  =&m+\lim_{T\to\infty}\lim_{z\to0}\frac{-i\vol\l(\d B_{\rn[3]}\r)\int_{\rn_{>0}}r^{z+3}\l(e^{iTr}-e^{-iTr}\r)dr}{\vol\l(\d B_{\rn[3]}\r)\int_{\rn_{>0}}r^{z+2}\l(e^{iTr}+e^{-iTr}\r)dr}\tag{$*$}\\
  =&m+\lim_{T\to\infty}\lim_{z\to0}\frac{\l(-e^{-i\frac{\pi (z+3)}{2}}-e^{-3i\frac{\pi (z+3)}{2}}\r)\Gamma(z+4)T^{-z-4}}{i\l(-e^{-i\frac{\pi (z+2)}{2}}+e^{-3i\frac{\pi (z+2)}{2}}\r)\Gamma(z+3)T^{-z-3}}\\
  =&m+\lim_{T\to\infty}\lim_{z\to0}\frac{\l(-e^{-i\frac{\pi (z+3)}{2}}-e^{-i\frac{\pi (z+3)}{2}}e^{-i\pi (z+3)}\r)\Gamma(z+4)T^{-z-4}}{i\l(-e^{-i\frac{\pi (z+2)}{2}}+e^{-i\frac{\pi (z+2)}{2}}e^{-i\pi (z+2)}\r)\Gamma(z+3)T^{-z-3}}\\
  =&m+\lim_{T\to\infty}\lim_{z\to0}\frac{e^{-i\frac{\pi (z+3)}{2}}\l(-1-e^{-i\pi (z+2)}e^{-i\pi}\r)(z+3)\Gamma(z+3)T^{-z-4}}{ie^{-i\frac{\pi (z+2)}{2}}\l(e^{-i\pi (z+2)}-1\r)\Gamma(z+3)T^{-z-3}}\\
  =&m+\lim_{T\to\infty}\lim_{z\to0}\frac{e^{-i\frac{\pi (z+2)}{2}}e^{-i\frac{\pi}{2}}(z+3)}{ie^{-i\frac{\pi (z+2)}{2}}T}\\
  =&m+\lim_{T\to\infty}\lim_{z\to0}\frac{-z-3}{T}\\
  =&m.
\end{align*}
In other words, we have correctly computed $E=mc^2$ again.

\begin{remark*}
  The calculation above highlights a number of properties which will be even more important in the case of the massive Schwinger model (section~\ref{sec:schwinger-gauge-boson-mass}). In the first summand, the ``observable'' $m$ depends on none of the variables which leads to many cancellations. This will be paramount for the massive Schwinger model since the part of the model that is not analytically solvable will vanish in one such cancellation.

  The other important property can be seen in the latter summand. Since the argument of the Laplace transform is homogeneous and the volume $T$ of the time-torus enters through the evaluation of the Laplace transform, we obtain that the limit $T\to\infty$ depends primarily on the asymptotic expansion of the observable. Thus, knowing the asymptotic behavior of the observable enables us to decide whether or not a term will vanish in the limit $T\to\infty$.
\end{remark*}

\begin{remark*}
  We should note that ($*$) can be implemented just like the implementations above since it is the Laplace transform
  \begin{align*}
    \Lp(r\mapsto r^q)(s)=\frac{\Gamma(q+1)}{s^{q+1}}
  \end{align*}
  which holds for $\Re(s)>0$ and $\Re(q)>-1$, and through analytic extension for $s\in\cn\setminus\{0\}$ and $q\in\cn\setminus(-\nn)$. In particular,
  \begin{align*}
    \int_{\rn_{>0}}r^ze^{iTr}dr=\frac{-ie^{-i\frac{\pi z}{2}}\Gamma(z+1)}{T^{z+1}}
  \end{align*}
  and
  \begin{align*}
    \int_{\rn_{>0}}r^ze^{-iTr}dr=\frac{ie^{-3i\frac{\pi z}{2}}\Gamma(z+1)}{T^{z+1}}.
  \end{align*}
  Furthermore, using stationary phase approximation (cf. chapter 8 in~\cite{hartung-phd}) in the setting of Theorem~\ref{theorem-regularization}, we can see that the regularization is given in terms of these Laplace transforms only. Hence, the actual difficulty in computing $\langle\Omega\rangle_T(z)$ are the integrals over $\d B_{\rn[N]}$ and possibly computing the limits $z\to0$ and $T\to\infty$.

\end{remark*}

Using these Laplace transforms and the fact that the integrals over $\d B_{\rn[N]}$ yield $\vol\d B_{\rn[N]}$ in this case, we obtain the following implementation for the ground state energy of a free relativistic fermion in $N$ spatial dimensions.
\begin{verbatim}
import sympy as smp

z,T = smp.symbols("z,T")
m,voldB,r = smp.symbols("m,voldB,r",positive=True)
N,k = smp.symbols("N,k",positive=True,integer=True)

f = voldB*k*r**(z+N-1)/(2*smp.pi)**N
g = voldB*k*smp.I*m*r**(z+N)/(2*smp.pi)**N

num = smp.laplace_transform(m*f,r,-smp.I*T)[0]
num += smp.laplace_transform(m*f,r,smp.I*T)[0]
num -= smp.laplace_transform(g,r,-smp.I*T)[0]
num -= smp.laplace_transform(g,r,smp.I*T)[0]

den = smp.laplace_transform(f,r,-smp.I*T)[0]
den += smp.laplace_transform(f,r,smp.I*T)[0]

lim = smp.limit(smp.limit(num/den,z,0),T,smp.oo)

print "<H_m> = "+str(smp.simplify(lim))
\end{verbatim}

\section{Gauge boson mass in the Schwinger model}\label{sec:schwinger-gauge-boson-mass}
At this point, we will return to the massive Schwinger model but add an abelian vector gauge field. Thus, the model becomes fully interacting with a non-trivial dynamics leading to the confinement of the charges and, hence, bound states. Hence, applying the $\zeta$-regularization of the generalized Kontsevich-Vishik trace constitutes a first highly non-trivial example of the proposed method. Here, we will only provide a demonstration for the calculation of the gauge boson mass. Further observables could be computed in a similar way if required.

Here, we have the fermionic Hamiltonian in the temporal gauge
\begin{align*}
  H_F=
  \begin{pmatrix}
    m&-i\d_1-eA\\
    -i\d_1-eA&m
  \end{pmatrix}
\end{align*}
as well as the self-interaction Hamiltonian
\begin{align*}
  H_S=-\frac{1}{4}F_{\mu\nu}F^{\mu\nu}=\frac{1}{2}E^2
\end{align*}
of the gauge field, where $F_{\mu\nu}=\d_\mu A_\nu-\d_\nu A_\mu$, $A=A_1$ ($A_0=0$ is the temporal gauge), and $E=-\d_0A$.

At this point, it is important to address the space-time dependence of the gauge fields. Due to the nature of the observable in question (the gauge boson mass), it is convenient to choose the family $((A(x),E(x)))_{x\in X}$ as canonical coordinates.\footnote{In a sense, this can be seen as a form of projective limit of discretized space.} Though the space torus $X$ needs to be formally introduced (recall that the $\zeta$-regularization needs a compact manifold), we will suppress it in the following since the limit $\vol(X)\to\infty$ is trivial. More importantly, this setting implies that the time-dependence of the gauge fields is implicit while the space-dependence is still explicit. 

Thus, 
\begin{align*}
  \sigma_{\exp\l(-i\int_0^T H\r)}=e^{-imT}e^{-\frac{i}{2}TE^2}
  \begin{pmatrix}
    \cos\l(T\xi-e\int_0^T A\r)&-i\sin\l(T\xi-e\int_0^T A\r)\\
    -i\sin\l(T\xi-e\int_0^T A\r)&\cos\l(T\xi-e\int_0^T A\r)
  \end{pmatrix}.
\end{align*}
In~\cite{schwinger}, Schwinger himself supplied us with the Green's function of the Abelian vector gauge field. From it, we can read off the observable $\Omega$ for the squared mass of the gauge boson;
\begin{align*}
  \sigma_\Omega=E^2+\frac{e^2}{\pi}.
\end{align*}
Hence, the gauge boson mass $m_g$ is given by (suppressing gauges for $\xi$, $x$, and $A$)
\begin{align*}
  \begin{aligned}
    m_g^2=&\langle\Omega\rangle\\
    =&\lim_{{T\to\infty}\atop{z\to 0}}\frac{\iiiint e^{-imT-\frac{i}{2}TE^2}\l(e^{iT\xi}e^{-ie\int_0^T A}+e^{-iT\xi}e^{ie\int_0^T A}\r)\l(E^2+\frac{e^2}{\pi}\r)\betr E^zd\xi dx dE\mathscr{D}A}{\iiiint e^{-imT-\frac{i}{2}TE^2}\l(e^{iT\xi}e^{-ie\int_0^T A}+e^{-iT\xi}e^{ie\int_0^T A}\r)\betr E^zd\xi dx dE\mathscr{D}A}\\
    =&\frac{e^2}{\pi}+\lim_{T\to\infty}\lim_{z\to 0}\frac{\iiiint e^{-\frac{i}{2}TE^2}\l(e^{iT\xi}e^{-ie\int_0^T A}+e^{-iT\xi}e^{ie\int_0^T A}\r)\betr E^{z+2}d\xi dxdE\mathscr{D}A}{\iiiint e^{-\frac{i}{2}TE^2}\l(e^{iT\xi}e^{-ie\int_0^T A}+e^{-iT\xi}e^{ie\int_0^T A}\r)\betr E^zd\xi dx dE\mathscr{D}A}\\
    =&\frac{e^2}{\pi}+\lim_{T\to\infty}\lim_{z\to 0}\frac{\int e^{-\frac{i}{2}TE^2}\betr E^{z+2}dE}{\int e^{-\frac{i}{2}TE^2}\betr E^z dE}\\
    =&\frac{e^2}{\pi}+\lim_{T\to\infty}\ubr{\lim_{z\to 0}\frac{-ie^{-\frac{3i\pi}{2}\frac{z+1}{2}}\Gamma\l(\frac{z+3}{2}\r)\l(\frac2T\r)^{\frac{z+3}{2}}}{-ie^{-\frac{3i\pi}{2}\frac{z-1}{2}}\Gamma\l(\frac{z+1}{2}\r)\l(\frac2T\r)^{\frac{z+1}{2}}}}_{\propto\frac1T}\\
    =&\frac{e^2}{\pi}.
  \end{aligned}
\end{align*}

\begin{remark*}
  This calculation highlights the cancellations and asymptotic properties we observed in section~\ref{sec:free-relativistic-Fermions} again. Here, the integrals with respect to $A$ are very difficult and not analytically solvable. However, due to the structure of the observable, these integrals cancel out to a factor of $1$. More importantly, even if they did not cancel, we would know that the second term had to vanish in the limit $T\to\infty$ since we know the asymptotics of the observable in $E$. More precisely, having $\abs E^{z+2}$ in the numerator and $\abs E^z$ in the denominator (and a phase function in terms of $E^2$) implies that the quotient is proportional to $\frac1T$ and, as such, vanishes for $T\to\infty$.
\end{remark*}

\section{Spontaneous symmetry breaking and mass - the $\phi^4$ model}\label{sec:spontaneous}
Since spontaneous symmetry breaking is essential to the Higgs mechanism, we will have a quick look at it here, as well. In the simplest relativistic case, we have scalar fields $\phi=(\phi^1,\ldots,\phi^k)$ and the Langrangian contains a potential term $V(\phi)$. Then, we are looking for constant fields $\phi_0^j$ which locally minimize $V$. These $\phi_0^j$ are the vacuum expectation values of the $\phi^j$. Furthermore, the matrix $\l(\d_{i}\d_{j}V(\phi_0)\r)_{i,j\in\nn_{\le k}}$ is symmetric and its eigenvalues give the squared masses of the fields. In particular, if $k=1$, we obtain the vacuum expectation values from $\d V(\phi_0)=0$ and $\d^2V(\phi_0)\ge0$ where $\sqrt{\d^2V(\phi_0)}$ is the mass of the field.

However, in general, we will not be able to simply read off $V(\phi)$. Instead, we will consider the partition function as a function of $\phi$ and obtain an effective potential $V_e(\phi)$ through the identity $Z(\phi)=\exp\l(-i\int V_e(\phi)d(t,x)\r)$, i.e.,
\begin{align*}
  V_e(\phi):=\frac{\ln Z(\phi)}{-iTX}
\end{align*}
where we used the fact that we are looking for constant $\phi$ and introduced a space-time torus of volume $TX$.

Consider the $\phi^4$ model whose Hamiltonian is given by
\begin{align*}
  H=\int\frac{p^2}{2}-\frac{1}{2}\phi\Delta\phi-\frac{1}{2}\mu^2\phi^2+\frac{\lambda}{4!}\phi^4dx.
\end{align*}
In this case, the minima are given by $\phi_0=\pm\sqrt{\frac6\lambda}\mu$ and the field mass is $\sqrt2\mu$. Using the $\zeta$-regularized partition function, we obtain
\begin{align*}
  Z(z,\phi)=&\frac{1}{2\pi}\int_{\rn}e^{-iTX\l(\frac{p^2}{2}-\frac{\mu^2}{2}\phi^2+\frac{\lambda}{4!}\phi^4\r)}\betr p^z dp
\end{align*}
which we may implement directly.
\begin{verbatim}
import sympy as smp

z = smp.symbols("z")
phi,p = smp.symbols("phi,p",real=True)
TX,mu,L = smp.symbols("TX,mu,L",positive=True)

H = p**2/2 - mu**2/2*phi**2 + L/24*phi**4
exph = smp.exp(-smp.I*TX*H)
gauge = smp.Abs(p)**z

Z = smp.fourier_transform(exph * gauge,p,0).doit()/(2*smp.pi)

V = smp.ln(Z)/(-smp.I*TX)
dV = smp.simplify(smp.diff(V,phi))
ddV = smp.diff(dV,phi)

# take limit z->0
dV = smp.limit(dV,z,0)
ddV = smp.limit(ddV,z,0)

extrema = smp.solve(dV,phi)

# check extrema for minima, in physical limit TX->\infty
ddV = smp.limit(ddV,TX,smp.oo)
for i in range(len(extrema)):
    extrema[i] = smp.limit(extrema[i],TX,smp.oo)

minima = []

for phi0 in extrema:
    if ddV.subs(phi,phi0)>=0: 
        minima.append(phi0)
        
print minima

masses = []
for phi0 in minima:
    m = smp.sqrt(ddV.subs(phi,phi0))
    if m not in masses:
        masses.append(m)

print masses
\end{verbatim}
Note that gauge independence of $Z(\phi):=\lim_{TX\to\infty}\lim_{z\to 0}Z(z,\phi)$ means this computation may only fail if $Z(z,\phi)$ has a pole in $0$ or $Z(\phi)=0$ since we cannot take the logarithm in that case ($Z(\phi)\in\cn\setminus\rn_{\ge0}$ can be treated choosing an appropriate branch cut of $\ln$). Here, neither of these cases occurs, i.e., the results are independent of the chosen gauge and we correctly obtain the minima $\pm\sqrt{\frac{6}{\lambda}}\mu$ and the field mass $\sqrt2\mu$.

\begin{remark*}
  The situation will be more complex if we are not using the fact that the $\phi$ are (spatially) constant since the $\phi\Delta\phi$ term will not vanish. In that case, it might be more appropriate to write the term as $-\langle\nabla\phi,\nabla\phi\rangle$ which is of the same form as the $p^2$ again, but the best choice will most likely depend on the specifics of the problem and observable in consideration. It should be noted, however, that space-discretization (i.e., replacing $\phi$ by a vector $(\phi(x_j))_j\in \rn[n]$) can be very viable and the resulting path integral will be $\zeta$-regularizable again.
\end{remark*}

\section*{Conclusion}
We proposed a new definition of the path integral based on Feynman's formulation (Definition~\ref{def:reg-path-int}). By construction the proposed definition restricts to Feynman's definition whenever it is well-defined; e.g., using Wick rotations, lattice discretization, or trace-class observables.

We obtained the new definition by replacing Feynman's path integral with its corresponding version as constructed using the generalized Kontsevich-Vishik trace. More precisely, we replaced the time-evolution semi-group $T(t):=e^{-\frac{i}{\hbar}\int_0^t H(s)ds}$ by a holomorphic family of operator families $z\mapsto \Gf(t,z)$ satisfying $\Gf(t,0)=T(t)$ and for which Feynman's path integral is well-defined if $\Re(z)$ is sufficiently small. If $\Gf$ is chosen appropriately, we showed that the path integrals defined for $\Re(z)\ll0$ can be extended meromorphically to $\cn$ (cf. Theorem~\ref{theorem-regularization} and Appendix~\ref{sec:KV}) and defined the regularized path integral as the value of the meromorphic extension at $z=0$ (provided it exists; Definition~\ref{def:reg-path-int}).

Furthermore, we considered a number of fundamental models, including the non-trivial cases of the massive Schwinger model and a $\phi^4$ theory, as evidence for the validity of the proposed definition and provided implementations using symbolic arithmetic. It is particularly important to note that the underlying regularization is a priori known which reduces the computational effort of evaluating these regularized path integrals to the evaluation of spherical integrals (and possibly the limits $z\to 0$ and $\mathrm{physical\ volume}\to\infty$). Hence, continuum computations without Wick rotations are possible with the new definition.

In particular, we can answer the questions we set out in the beginning.
\begin{itemize}
\item[(i)] Is it possible to $\zeta$-regularize the partition function and expectation values of observables in Minkowski space-time?
\end{itemize}

Yes, provided that the Hamiltonian and observable satisfy certain homogeneity and positivity or non-degeneracy assumptions in the leading order terms (cf. Theorem~\ref{theorem-regularization}). 

\begin{itemize}
\item[(ii)] Are the regularized partition functions and expectation values of observables independent of the choices made in the construction of the $\zeta$-function?
\end{itemize}

Almost always, the answer to this question is yes. Both are well-defined and choice independent if there are no critical degrees of homogeneity (which depends only on the space-time dimension). Thus, the quotient is almost always well-defined and independent of the choices made (though it might be infinite). Gauge dependence can only appear if the Laurent coefficient of lowest possible order vanish (which depends on the degrees of homogeneity and logarithmic degrees at critical degree of homogeneity). In particular, if there are no critical degrees of homogeneity, then the partition function is gauge independent and the expectation value of the observable can only depend on the gauge if we are in the $\frac00$ case.\footnote{Note that being in the $\frac00$ case is gauge independent, i.e., the $\frac00$ case cannot be removed through the choice of gauge.} On the other hand, this case does appear in practice; the ground state energy of the free relativistic fermion (section~\ref{sec:free-relativistic-Fermions}), for instance, is of this form. To overcome the problem that different choices of gauge in numerator and denominator can generate arbitrary results, we conjecture that choosing the same gauge should give a meaningful choice physically.

\begin{itemize}
\item[(iii)] Does the regularization contain the known special cases of well-defined path integrals (e.g., Wick rotated or space-time discretized)?
\end{itemize}

Yes. Space-time discretization replaces the operators by matrices. Hence all traces are well-defined and the construction of the $\zeta$-regularization coincides with the canonical trace on trace-class operators. Similarly, Wick rotations yield pseudo-differential operators and there it is known that the $\zeta$-regularization used in this context is the unique extension of the canonical trace.

\begin{itemize}
\item[(iv)] Is the construction physically ``meaningful''?
\end{itemize}

This question can be interpreted in different ways. On one hand, we may ask if the regularization can be interpreted physically. In this sense, choosing the same gauge for both numerator and denominator in the expectation value of observables is important. This means that we replace the time-evolution of our system by a holomorphic family of time-evolutions. In other words, we consider a holomorphic family of physical systems and conjecture that the physical values of the system to be studied can be obtained by analytic continuation.

On the other hand, we need to ask whether or not the regularized theory is physically correct. By construction, we know that the regularized theory coincides with the physical theory if we have trace-class operators to begin with. In case the regularization is necessary, we have considered a number of physical models. In each of these models, the regularization not only recovered the known physical values but is also computable. In fact, the dependence on the regularizing parameter is known explicitly and the remaining integrals are over compact manifolds.

\begin{appendix}
\section{The generalized Kontsevich-Vishik trace}\label{sec:KV}
In this appendix, we will give a non-technical overview of Fourier Integral Operator $\zeta$-functions and the (generalized) Kontsevich-Vishik trace. For more detail, please refer to~\cite{hartung-phd,hartung-scott}.

Given a closed, compact, orientable, connected, finite dimensional Riemannian manifold $X$ and a closed conic Lagrangian submanifold $\Lambda$ of $T^*X^2\setminus0$, we can consider the space $I^m(X^2;\Lambda)$ of Lagrangian distributions of order $m$ with microsupport in $\Lambda$ (cf. Chapter 25 in~\cite{hoermander-books}). Integral operators with kernels in some $I^m(X^2;\Lambda)$ are called Fourier Integral Operators. More precisely, we have the following definition.

\begin{definition}
  Let $X$ be a $C^\infty$ manifold, $E$ a (complex) vector bundle over $X$, and $Y$ a closed $C^\infty$ sub-manifold of $X$. Then, the space $I^m(X,Y;E)$ of distribution sections of $E$ that are conormal to $Y$ and of order less than or equal to $m$ is the set of all distributions $u\in C_c^\infty(X,E)'$ such that 
  \begin{align*}
    L_1\ldots L_Nu\in B_{2,\infty,\loc}^{-m-\frac{\dim X}{4}}(X,E)
  \end{align*}
  for all $N\in\nn_0$ and all first order differential operators $L_j$ between distribution sections of $E$ whose coefficients are $C^\infty$ tangential to Y.
\end{definition}

\begin{remark*}
  Here, $B_{p,q}^s(\rn[n])$ denotes the usual Besov space and, for $U\sse\rn[n]$ open, we define $B_{p,q,\loc}^s(U)$ as the set of distributions $u\in C_c^\infty(U)'$ such that $\fa \phi\in C_c^\infty(U):\ \phi u\in B_{p,q}^s(\rn[n])$. This definition can then be lifted to manifolds in the usual manner.
\end{remark*}

The definition of conormality can be extended to pseudo-differential operators from $E$ to $E$ with principal symbol vanishing on $Y$. Thus, it can be extended to Lagrangian manifolds.

\begin{definition}
  Let $X$ be a $C^\infty$ manifold, $E$ a (complex) vector bundle over $X$, and $\Lambda\sse T^*X\setminus 0$ a closed, conic, $C^\infty$, Lagrangian sub-manifold. Then, the space $I^m(X,\Lambda;E)$ of Lagrangian distribution sections of $E$ of order less than or equal to $m$ is the set of all distributions $u\in C_c^\infty(X,E)'$ such that 
  \begin{align*}
    L_1\ldots L_Nu\in B_{2,\infty,\loc}^{-m-\frac{\dim X}{4}}(X,E)
  \end{align*}
  for all $N\in\nn_0$ and all properly supported first order pseudo-differential operators $L_j\in\Psi^1(X;E,E)$ whose principal symbols vanish on $\Lambda$.
\end{definition}

It is common to denote $\Lambda$ in terms of a canonical relation $\Gamma\sse\l(T^*X\setminus0\r)^2$ (cf., e.g., Chapter 1 in~\cite{hartung-phd}) which satisfies
\begin{align*}
  \Lambda=\Gamma':=\l\{((x,\xi),(y,\eta))\in\l(T^*X\setminus0\r)^2;\ ((x,\xi),(y,-\eta))\in\Gamma\r\}.
\end{align*}
If $\Gamma$ is chosen to be a homogeneous canonical relation (cf., e.g., Chapter 1 in~\cite{hartung-phd}, Theorem 2.4.1 in~\cite{duistermaat}, and Example 1 in~\cite{guillemin-residues}), then the set of operators $\Ap_\Gamma$ with kernels in $\bigcup_{m\in\rn}I^m(X^2;\Gamma')$ forms an associative algebra. Furthermore, it can be shown that $\Ap_\Gamma$ has a non-trivial intersection with the set of trace-class operators in $L(L_2(X))$ (cf. Lemmata 1.12 and 1.13 in~\cite{hartung-phd}); more precisely, if $A\in\Ap_\Gamma$ has kernel\footnote{An integral operator $A$ has kernel $k$ if and only if $Af(x)=\int k(x,y)f(y)dy$ holds for all $f$ in the domain of $A$.} $k\in I^m(X^2;\Gamma')$ with $m$ sufficiently small, then $A$ is of trace-class, $k$ continuous, and
\begin{align*}
  \tr A=\int_Xk(x,x)d\vol_X(x).
\end{align*}
In many applications (like the Feynman path integral; cf. section~\ref{sec:feynman}) we would like to extend this trace to operators that are not of trace-class. Such an extension of the trace can be obtained using $\zeta$-regularization. Let $A_0\in\Ap_\Gamma$ with kernel $k_0\in I^m(X^2;\Gamma')$. Then, we consider a holomorphic family $A\in C^\omega\l(\cn,\Ap_\Gamma\r)$ with kernels $k\in C^\omega\l(\cn,I^m(X^2;\Gamma')\r)$ such that $k(0)=k_0$ and $\fa z\in\cn:\ k(z)\in I^{m+\Re(z)}(X;\Gamma')$, and define $\zeta(A)$ to be the maximal meromorphic extension of
\begin{align*}
  \zeta(A)(z):= \tr A(z)=\int_Xk(z)(x,x)d\vol_X(x)
\end{align*}
which is well-defined for $\Re(z)$ sufficiently small, that is, $\Re(z)\ll0$.

An important class of holomorphic families of Fourier Integral Operators are gauged Fourier Integral Operators with $\log$-polyhomogeneous amplitudes. These gauged Fourier Integral Operators have kernels of the form
\begin{align*}
  k(z)(x,y)=\int_{\rn[N]}e^{i\theta(x,y,\xi)}a(z)(x,y,\xi)d\xi
\end{align*}
where $\theta$ is a phase function 
\begin{align*}
  \theta(x,y,\xi)=\theta\l(x,y,\frac{\xi}{\norm\xi_{\ell_2(N)}}\r)\norm\xi_{\ell_2(N)}
\end{align*}
and
\begin{align*}
  a(z)(x,y,\xi)=a_0(z)(x,y,\xi)+\sum_{\iota\in I}a_\iota(z)(x,y,\xi)
\end{align*}
where $a_0(z)\in L_1(X\times X\times\rn[N])$ and
\begin{align*}
  a_\iota(z)(x,y,\xi)=\norm\xi_{\ell_2(N)}^{d_\iota+z}\l(\ln\norm\xi_{\ell_2(N)}\r)^{l_\iota}\tilde a_\iota\l(x,y,\frac{\xi}{\norm\xi_{\ell_2(N)}}\r)
\end{align*}
holds with a number of additional properties making everything well-defined (cf. Chapter 2 in~\cite{hartung-phd}). We call $d_\iota$ the degree of homogeneity of $a_\iota$ and $l_\iota$ the logarithmic order. If all $l_\iota$ vanish, then we call the amplitude polyhomogeneous.

Fourier Integral Operator $\zeta$-functions with polyhomogeneous amplitudes were shown to exist as meromorphic functions on $\cn$ and had their residues studied by Guillemin~\cite{guillemin-lagrangian,guillemin-residues}. Using Guillemin's approach and introducing the notion of gauged poly-$\log$-homogeneous distributions, the author~\cite{hartung-phd,hartung-scott} was able to compute the Laurent expansion of $\zeta(A)$ for Fourier Integral Operators with $\log$-polyhomogeneous amplitudes, as well. In particular, it can be shown that $\zeta(A)$ only has isolated poles of finite order. The poles are located at $z=-N-d_\iota$ and the maximal pole order is $l_\iota+1$.

While the residues of $\zeta(A)$ yield important traces (cf., e.g.,~\cite{guillemin-residues}), we are interested in the values $\zeta(A)(0)$ provided none of the degrees of homogeneity satisfies $d_\iota=-N$. Then, $\zeta(A)$ is holomorphic in a neighborhood of zero and $\zeta(A)(0)$ depends only on $A(0)$. In fact,
\begin{align*}
  A_0\ \rightsquigarrow\ A\in C^\omega(\cn,\Ap_\Gamma)\text{ gauged with }A(0)=A_0\ \rightsquigarrow\ \zeta(A)(0)
\end{align*}
defines a trace provided the amplitude of $A_0$ has no critical degree of homogeneity $d_\iota=-N$ (cf. Chapter 7 in~\cite{hartung-phd}). This trace is called the generalized Kontsevich-Vishik trace (the original Kontsevich-Vishik trace is the special case of $A_0$ being a classical pseudo-differential operator) and given by
\begin{align*}
    \begin{aligned}
      \zeta(A)(0)=&\int_X\int_{B_{\rn[N]}(0,1)}e^{i\theta(x,x,\xi)}a(0)(x,x,\xi)\ d\xi\ d\vol_X(x)\\
      &+\int_{\rn_{\ge1}\times\d B_{\rn[N]}}\int_Xe^{i\theta(x,x,\xi)}a_0(0)(x,x,\xi)\ d\vol_X(x)\ d\vol_{\rn_{\ge1}\times\d B_{\rn[N]}}(\xi)\\
      &+\sum_{\iota\in I}\frac{(-1)^{l_\iota+1}l_\iota!\int_{X\times\d B_{\rn[N]}}e^{i\theta(x,x,\xi)}\tilde a_\iota(0)(x,x,\xi)\ d\vol_{X\times\d B_{\rn[N]}}(x,x,\xi)}{(N+d_\iota)^{l_\iota+1}}.
    \end{aligned}
\end{align*}

By construction, the generalized Kontsevich-Vishik trace coincides with $\tr$ on trace-class operators. Furthermore, it was shown that the Kontsevich-Vishik trace is the only trace on the algebra of classical pseudo-differential operators that restricts to $\tr$ in $L(L_2(X))$~\cite{maniccia-schrohe-seiler}. These properties make the generalized Kontsevich-Vishik trace a prime candidate for path integral regularization as such a path integral regularization is consistent with respect to discretization (turning operators into matrices and, thus, trace-class) and Wick rotations (turning the path integral into pseudo-differential operator traces).  
\end{appendix}


\begin{bibdiv}
  \begin{biblist}
    \bib{jansen-cichy}{article}{
      author={BA\~{N}ULS, M. C.},
      author={CICHY, K.},
      author={CIRAC, J. I.},
      author={JANSEN, K.},
      title={The mass spectrum of the Schwinger model with Matrix Product States},
      journal={Journal of High Energy Physics},
      volume={158},
      date={2013}
    }
    \bib{creutz-freedman}{article}{
      author={CREUTZ, M.},
      author={FREEDMAN, B.},
      title={A Statistical Approach to Quantum Mechanics},
      journal={Annals of Physics},
      volume={132},
      pages={427-462},
      date={1981}
    }
    \bib{degrand-detar}{book}{
      author={DEGRAND, T.},
      author={DETAR, C.},
      title={Lattice Methods for Quantum Chromodynamics},
      publisher={World Scientific},
      address={Singapore},
      date={2006}
    }
    \bib{duistermaat}{book}{
      author={DUISTERMAAT, J. J.},
      title={Fourier Integral Operators},
      publisher={Birkh\"{a}user},
      date={1996}
    }
    \bib{feynman}{article}{
      author={FEYNMAN, R. P.},
      title={Space-Time Approach to Non-Relativistic Quantum Mechanics},
      journal={Reviews of Modern Physics},
      volume={20},
      pages={367-387},
      date={1948}
    }
    \bib{feynman-hibbs-styer}{book}{
      author={FEYNMAN, R. P.},
      author={HIBBS, A. R.},
      author={STYER, D. F.},
      title={Quantum Mechanics and Path Integrals},
      publisher={Dover Publications, Inc.},
      edition={Emended Edition},
      address={Mineola, NY},
      date={2005}
    }
    \bib{gattringer-lang}{book}{
      author={GATTRINGER, C.},
      author={LANG, C. B.},
      title={Quantum Chromodynamics on the Lattice},
      publisher={Springer},
      address={Berlin/Heidelberg},
      date={2010}
    }
    \bib{gibbons-hawking-perry}{article}{
      author={GIBBONS, G. W.},
      author={HAWKING, S. W.},
      author={PERRY, M. J.},
      title={Path integrals and the indefiniteness of the gravitational action},
      journal={Nuclear Physics},
      volume={B138},
      pages={141-150},
      date={1978}
    }
    \bib{guillemin-lagrangian}{article}{
      author={GUILLEMIN, V.},
      title={Gauged Lagrangian Distributions},
      journal={Advances in Mathematics},
      volume={102},
      pages={184-201},
      date={1993}
    }
    \bib{guillemin-residues}{article}{
      author={GUILLEMIN, V.},
      title={Residue Traces for certain Algebras of Fourier Integral Operators},
      journal={Journal of Functional Analysis},
      volume={115},
      pages={391-417},
      date={1993}
    }
    \bib{hartung-phd}{book}{
      author={HARTUNG, T.},
      title={$\zeta$-functions of Fourier Integral Operators},
      publisher={Ph.D. thesis, King's College London},
      address={London},
      date={2015}
    }
    \bib{hartung-scott}{article}{
      author={HARTUNG, T.},
      author={SCOTT, S.},
      title={A generalized Kontsevich-Vishik trace for Fourier Integral Operators and the Laurent expansion of $\zeta$-functions},
      journal={arXiv:1510.07324v2~[math.AP]},
    }
    \bib{hawking}{article}{
      author={HAWKING, S. W.},
      title={Zeta Function Regularization of Path Integrals in Curved Spacetime},
      journal={Communications in Mathematical Physics},
      volume={55},
      pages={133-148},
      date={1977}
    }
    \bib{hoermander-books}{book}{
      author={H\"{O}RMANDER, L.},
      title={The Analysis of Linear Partial Differential Operators},
      part={I-IV},
      publisher={Springer},
      address={Berlin/Heidelberg},
      date={1990}
    }
    \bib{johnson-freyd}{article}{
      author={JOHNSON-FREYD, T.},
      title={The formal path integral and quantum mechanics},
      journal={Journal of Mathematical Physics},
      volume={51},
      date={2010}
    }
    \bib{kontsevich-vishik}{article}{
      author={KONTSEVICH, M.},
      author={VISHIK, S.},
     title={Determinants of elliptic pseudo-differential operators},
      journal={Max Planck Preprint, arXiv:hep-th/9404046},
      date={1994}
    }
    \bib{kontsevich-vishik-geometry}{article}{
      author={KONTSEVICH, M.},
      author={VISHIK, S.},
      title={Geometry of determinants of elliptic operators},
      journal={Functional Analysis on the Eve of the XXI century, Vol. I, Progress in Mathematics},
      volume={131},
      pages={173-197},
      date={1994}
    }
    \bib{kumano-go-I}{article}{
      author={KUMANO-GO, N.},
      title={Phase space Feynman path integrals with smooth functional derivatives by time slicing approximation},
      journal={Bulletin des sciences mathematiques},
      volume={135},
      pages={936-987},
      date={2011}
    }
    \bib{kumano-go-II}{article}{
      author={KUMANO-GO, N.},
      author={FUJIWARA, D.},
      title={Phase space Feynman path integrals via piecewise bicharacteristic paths and their semiclassical approximations},
      journal={Bulletin des sciences mathematiques},
      volume={132},
      pages={313-357},
      date={2008}
    }
    \bib{kumano-go-III}{article}{
      author={KUMANO-GO, N.},
      author={VASUDEVA MURTHY, A. S.},
      title={Phase space Feynman path integrals of higher order parabolic type with general functional as integrand},
      journal={Bulletin des sciences mathematiques},
      volume={139},
      pages={495-537},
      date={2015}
    }
    \bib{maniccia-schrohe-seiler}{article}{
      author={MANICCIA, L.},
      author={SCHROHE, E.},
      author={SEILER, J.},
      title={Uniqueness of the Kontsevich-Vishik trace},
      journal={Proceedings of the American Mathematical Society},
      volume={136 (2)},
      pages={747-752},
      date={2008}
    }
    \bib{montvay-muenster}{book}{
      author={MONTVAY, I.},
      author={M\"{U}NSTER, G.},
      title={Quantum Fields on a Lattice},
      publisher={Cambridge University Press},
      address={Cambridge},
      date={1994}
    }
    \bib{paycha}{article}{
      author={PAYCHA, S.},
      title={Zeta-regularized traces versus the Wodzicki residue as tools in quantum field theory and infinite dimensional geometry},
      journal={Proceedings of the International Conference on Stochastic Analysis and Applications},
      pages={69-84},
      date={2001}
    }
    \bib{ray}{article}{
      author={RAY, D. B.},
      title={Reidemeister torsion and the Laplacian on lense spaces},
      journal={Advances in Mathematics},
      volume={4},
      pages={109-126},
      date={1970}
    }
    \bib{ray-singer}{article}{
      author={RAY, D. B.},
      author={SINGER, I. M.},
      title={$R$-torsion and the Laplacian on Riemannian manifolds},
      journal={Advances in Mathematics},
      volume={7},
      pages={145-210},
      date={1971}
    }
    \bib{schwinger}{article}{
      author={SCHWINGER, J.},
      title={Gauge Invariance and Mass II},
      journal={Physical Review},
      volume={128},
      pages={2425-2429},
      date={Dec. 1962}
    }
    \bib{takhtajan}{book}{
      author={TAKHTAJAN, L. A.},
      title={Quantum Mechanics for Mathematicians},
      publisher={American Mathematical Society},
      address={Providence, RI},
      date={2008}
    }
    \bib{wilson}{article}{
      author={WILSON, K. G.},
      title={Confinement of quarks},
      journal={Physical Review D},
      volume={10},
      pages={2445},
      date={1974}
    }
  \end{biblist}
\end{bibdiv}
\end{document}